\documentclass[12pt,draftcls,onecolumn]{IEEEtran}
\usepackage{amsmath}
\usepackage{amssymb}
\ifCLASSINFOpdf
   \usepackage[pdftex]{graphicx}
\else
   \usepackage[dvips]{graphicx}
\fi
\usepackage{algorithm}
\usepackage{algorithmic}

\usepackage{array}
\usepackage{fancyhdr}
\begin{document}

\newtheorem{remark}{Remark}
\newtheorem{proposition}{\textit{Proposition}}
\newtheorem{lemma}{\textit{Lemma}}
\newtheorem{theorem}{\textit{Theorem}}

\newtheorem{definition}{\textit{Definition}}
\newtheorem{problem}{\textit{Problem}}

\title{Access Policy Design for Cognitive Secondary Users under a Primary Type-I HARQ Process }

\author{Roghayeh~Joda,~\IEEEmembership{~Member,~IEEE}
         and Michele~Zorzi,~\IEEEmembership{Fellow,~IEEE}
\thanks{%
R. Joda is with the Department of Communication Technology, Iran
Telecommunication Research Center, Tehran, Iran (e-mail:
{r.joda}@itrc.ac.ir)}
\thanks{%
M. Zorzi is with the Department of Information Engineering,
University of Padova, Italy (e-mail: {zorzi}@dei.unipd.it).}
\thanks{%
This work is a generalization of \cite{R22} presented at the IEEE
International Conference on Communications-Cooperative and Cognitive
Mobile Networks Workshop in June 2014.}}

\maketitle \thispagestyle{empty}
\IEEEpeerreviewmaketitle
\begin{abstract}
In this paper, an underlay cognitive radio network that consists of
an arbitrary number of secondary users (SU) is considered, in which
the primary user (PU) employs Type-I Hybrid Automatic Repeat Request
(HARQ). Exploiting the redundancy in PU retransmissions, each SU
receiver applies forward interference cancelation to remove a
successfully decoded PU message in the subsequent PU
retransmissions. The knowledge of the PU message state at the SU
receivers and the ACK/NACK message from the PU receiver are sent
back to the transmitters. With this approach and using a Constrained
Markov Decision Process (CMDP) model and Constrained Multi-agent MDP
(CMMDP), centralized and decentralized optimum access policies for
 SUs are proposed to maximize their average sum throughput under
a PU throughput constraint. In the decentralized case, the channel
access decision of each SU is unknown to the other SU. Numerical
results demonstrate the benefits of the proposed policies in terms
of sum throughput of SUs. The results also reveal that the
centralized access policy design outperforms the decentralized
design especially when the PU can tolerate a low average long term
throughput. Finally, the difficulties in decentralized access policy
design with partial state information are discussed.
\end{abstract}

\section{Introduction}
The advent of new technologies and services in wireless
communication has increased the demand for spectrum resources so
that the traditional fixed frequency allocation will not be able to
meet these bandwidth requirements. However, most of the spectrum
frequencies assigned to licensed users are under-utilized. Thus,
cognitive radio is proposed to improve the spectral efficiency of
wireless networks \cite{R1}. Cognitive radio enables licensed
primary users (PUs) and unlicensed secondary users (SUs) to coexist
and transmit in the same frequency band \cite{R2}, \cite{R3}. For a
literature review on spectrum sharing and cognitive radio, the
reader is referred to \cite{R4}\nocite{R5}-\cite{R6}. In the
underlay cognitive radio approach, the smart SUs are allowed to
simultaneously transmit in the licensed frequency band allotted to
the PU. The PU is oblivious to the presence of the SUs while the SU
needs to control the interference it causes at the PU receiver.

HARQ, a link layer mechanism, is a combination of high-rate forward
error-correcting coding (FEC) and ARQ error-control, and is employed
in current technologies, including for example HSDPA and LTE. CRNs
with an HARQ scheme implemented by the PU are addressed in
\cite{R7}\nocite{R9}\nocite{R11}\nocite{R7}\nocite{R8}\nocite{R151}\nocite{R12}-\nocite{R181}\cite{R27}.
\cite{R7}, \cite{R9} and \cite{R11} show how to exploit the Type-I
HARQ retransmissions implemented by the PU. \cite{R7} considers a
cognitive radio network composed of one PU and one SU, and does not
utilize interference cancelation (IC) at the SU receiver. \cite{R9}
employs Type-I HARQ with an arbitrary number of retransmissions and
applies backward and forward IC after decoding the PU message at the
SU receiver. The network considered in \cite{R11} is similar to
\cite{R7}, where the SU is also allowed to selectively retransmit
its own previous corrupted message and apply a chain decoding
protocol to derive the SU access policy. \cite{R8} applies Type-II
Hybrid ARQ with at most one retransmission, where the SU receiver
tries to decode the PU message in the first time slot and, if
successful, it removes this PU message in the second time slot to
improve the SU throughput. The extension of the work in \cite{R8} to
IR-HARQ with multiple rounds is addressed in \cite{R151}, where
several schemes are proposed. \cite{R12} proposes SU transmission
schemes when the SU is able to infrequently probe the channel using
the PU Type-II HARQ feedback with Chase combining (CC-HARQ).
Exploiting primary Type-II HARQ in CRN has also been studied in
\cite{R181} and \cite{R27}. Note that deriving the benefit from PU
Type-I HARQ for designing an optimum access policy has been only
addressed for CRNs with one SU in the literature, with the exception
of our work in \cite{R22}. We have to notice that increasing the
number of SUs and allowing them to access the channel cause more
interference at the PU receiver and therefore decrease the PU
throughput. In fact it is necessary to control the access of the SUs
to the channel to constrain the PU throughput degradation.

In this paper, an optimum access policy for $N$ SUs is designed,
which exploits the redundancy introduced by the Type-I HARQ protocol
in transmitting copies of the same PU message and interference
cancelation at the SU receivers. The aim is to maximize the average
long term sum throughput of SUs under a constraint on the average
long term PU throughput degradation. We assume that the number of
transmissions is limited to at most $T$ and all SUs have a new
packet to transmit in each time slot. Two design scenarios are
considered: in the first one, SUs make a channel access decision
jointly, whereas in the second scenario, each SU makes an
independent decision and does not know whether or not the other
secondary users access the channel. We call them respectively as
centralized and decentralized scenarios. Noting the PU message
knowledge state at each of the SU receivers and also the ARQ
retransmission time, the $PU-SU_1-...-SU_N$ network is modeled using
MDP and MMDP models \cite{R23}, respectively in centralized and
decentralized scenarios. Due to the constraint on the average long
term PU throughput, we then have a constrained MDP (CMDP) and
Constrained MMDP (CMMDP).

In the centralized case, the access policy in one state shows the
joint probability of accessing and/or not accessing the channel by
the SUs. Using \cite{R10} and \cite{R13}, it follows that the
optimal policy may be obtained from the solution of a corresponding
LP problem. In the decentralized scenario, there is an access policy
for each SU describing the probability of accessing the channel by
that SU. It is noteworthy that we are interested in random access
policies instead of only deterministic access policies. Hence, the
optimum polices in the centralized case can not be directly applied
to a decentralized scenario. To propose local optimum access
policies for the CMMDP model, we employ Nash Equilibrium.

The simulation results demonstrate that due to the use of forward IC
(FIC), a cognitive radio network converges to the upper bound faster
as the number of SUs increases for large enough SNR of the channels
from the PU transmitter to SU receivers. The results also reveal
that our proposed centralized access policy design significantly
outperforms the decentralized one when the average PU throughput
constraint is low.

The paper is organized as follows. Following the system model in
Section \ref{system model}, the rates and the corresponding outage
probabilities are computed in Section \ref{Rate}. Optimal access
policies for $N$ SUs in centralized and decentralized scenarios are
proposed respectively in Sections \ref{SIII} and \ref{S_DeC}. The
numerical results are presented in Section \ref{SIV} and an
extension to the paper is discussed in Section \ref{SV}. Finally,
the paper is concluded in Section \ref{SVI}.

\section{System Model}\label{system model}
In the system we consider, there exist one primary and N secondary
transmitters denoted by $PU_{tx}$, $SU_{tx1}$,...,$SU_{txN}$,
respectively. These transmitters transmit their messages with
constant power over block fading channels. In each time slot (one
block of the channel), the channels are considered to be constant.
The instantaneous signal to noise ratios of the channels
$PU_{tx}\rightarrow PU_{rx}$,  $PU_{tx}\rightarrow SU_{rxn}$,
$SU_{txn}\rightarrow SU_{rxm}$, $SU_{txn}\rightarrow PU_{rx}$,
$n,m\in\{1,...,N\}$ are denoted by $\gamma_{pp}$, $\gamma_{ps_n}$,
$\gamma_{s_ns_m}$ and $\gamma_{s_np}$, respectively. As an example,
the system model with the mentioned channel SNRs for $N=2$ is
depicted in Fig. \ref{Fsys}.
\begin{figure}[t!]
\centering
\includegraphics[width=2.2 in]{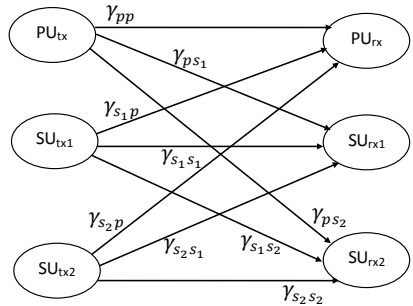}\caption{CRN Model with two SUs}
\label{Fsys}
\end{figure}

We assume that no Channel State Information (CSI) is available at
the transmitters except the ACK/NACK message and the PU message
knowledge state. Thus, transmissions are under outage, when the
selected rates are greater than the current channel capacity.

PU is unaware of the presence of the SUs and employs Type-I HARQ
with at most $T$ transmissions of the same PU message. We assume
that the ARQ feedback is received by the PU transmitter at the end
of a time-slot and a retransmission can be performed in the next
time-slot. Retransmission of the PU message is performed if it is
not successfully decoded at the PU receiver until the PU message is
correctly decoded or the maximum number of transmissions allowed,
$T$, is reached \footnote{A different type of HARQ, namely Type-II,
successively transmits incremental redundancy for the same packet
until success or until the maximum number of transmissions is
reached. While HARQ Type-II is out of the scope of the present
paper, we refer the interested reader to \cite{R27} for an initial
study and some preliminary results.}. Fig. \ref{ARQ} shows the model
of the PU Type-I HARQ, where $R_P$ is the PU transmission rate and
$C_t$ is the capacity of the $PU_{tx}$ to $PU_{rx}$ channel in ARQ
time slot $t$ when SU transmissions are considered as background
noise at $PU_{rx}$. In each time-slot, each SU, if it accesses the
channel, transmits its own message, otherwise it stays idle and does
not transmit. This decision is based on the access policy described
later. The activity of the SUs affects the outage performance of the
PU, by creating interference at the PU receiver. The objective is to
design access policies for SUs to maximize the average sum
throughput of the SUs under a constraint on the PU average
throughput degradation.

\begin{figure}[t!]
\centering
\includegraphics[width=5 in]{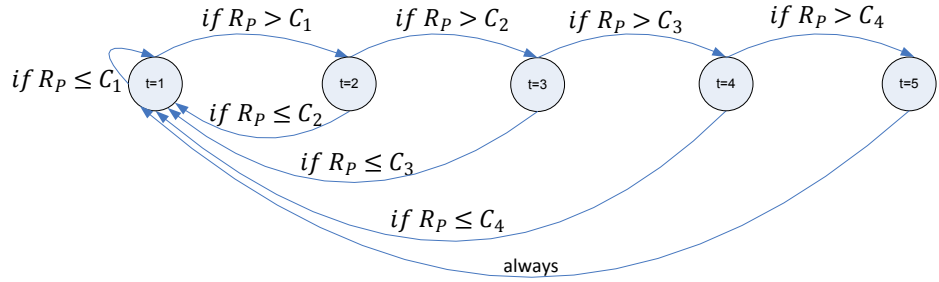}\caption{ARQ Type-I HARQ Model for $T=5$.}
\label{ARQ}
\end{figure}

We consider centralized and decentralized scenarios. In the
centralized scenario, there exists a central unit which receives the
PU message knowledge states of the SUs as well as the ACK/NACK
message from the PU receiver. This unit then computes the secondary
access actions and provides them to the SUs. In the decentralized
scenario, there exists no central unit. The PU message knowledge
state at each SU receiver is fed back to all the SU transmitters,
but each SU transmitter makes its own channel access decision
independently, based on this information. Thus, in the decentralized
design each SU is not aware of the access decisions of the other SUs
in the same slot.

If $SU_{rxn}$, $n\in\{1,...,N\}$, succeeds in decoding the PU
message, it can cancel it from the received signal in future
retransmissions. We refer to this as FIC \cite{R9}. We call the PU
message knowledge state as $\Phi=\bigl(\phi(1),...,\phi(N)\bigr)$,
which belongs to the set of $2^N$ possible combinations of PU
message knowledge states of all secondary users, where $\phi(n)$ is
the PU message knowledge state of the $SU_n$ receiver. For example,
if $\Phi=(K,K)$ for $N=2$, then $SU_{rx1}$ and $SU_{rx2}$ both know
the PU message and thus can perform FIC.

In the centralized scenario, there are $2^N$ possible channel access
combinations for the $N$ SUs, collected in the set
$A=\{0,1,...,2^N-1\}$. Each action, denoted by $a$, can be
represented as an $N$-dimensional vector
$\Psi(a)=\bigl(\varphi(a,1),...,\varphi(a,N)\bigr)$ which is equal
to the binary expansion of $a$, $0\leq a\leq2^N-1$ and therefore,
$\varphi(a,n)\in\{0,1\}$. Equivalently, we have
\begin{align}\label{DtoB}
\Psi(a)=Dec2Bin_N(a),
\end{align}
where the function $Dec2Bin_N$ is the $N$-dimensional decimal to
binary conversion. For access action $a$, $\varphi(a,n)=1$ means
that $SU_n$ is allowed to access the channel. If
$\Psi(a)=\mathbf{U}_n$, only $SU_n$ accesses the channel, where
$\mathbf{U}_n$ is defined as follows:
\begin{definition}
$\mathbf{U}_n$ is an $N$-dimensional vector with $\mathbf{U}_n(n)=1$
and $\mathbf{U}_n(m)=0$ for $m\neq n$.
\end{definition}

On the contrary, in the decentralized case, the access action is
$a_n\in A_n=\{0,1\}$ for secondary user $n$, where $a_n=1$ means
that this user is allowed to transmit.

\section{Rates and Outage Probabilities }\label{Rate}
First we consider the centralized scenario, where we have a joint
access action $a\in A=\{0,1,...,2^N-1\}$ and then we address the
decentralized scenario with independent $N$  access actions $a_n\in
A_n=\{0,1\}$, $n\in{1,...,N}$.

\subsection{Centralized Scenario}\label{Rate-Cent}
 The PU transmission rate, $R_P$, is considered fixed. However, based on the PU
message knowledge state $\Phi$ and the access action $a$, the rate
of each secondary user $n$ can be adapted and is denoted by
$R_{s_n,a,\Phi}$, $a\in A=\{1,...,2^N-1\}$. (All rates for access
action $a=0$ are zero.)

The outage probability of the channel $PU_{tx}\rightarrow PU_{rx}$
for SU access action $a$ is denoted by $\rho_{p,a}$. Noting that the
$SU_n$ transmissions $\forall n\in\{1,...,N\}$ are considered as
background noise at $PU_{rx}$, we have
\begin{align}\label{E34}
\rho_{p,a}&=1-Pr\left(R_p\leq
C(\frac{\gamma_{pp}}{1+\sum_{n=1}^{N}\varphi(a,n)\gamma_{s_np}})\right)\;\;\;\;a\in
A=\{0,1,...,2^N-1\},
\end{align}
where $C(x)=\log_2(1+x)$. Obviously, $C_t$ in Fig. \ref{ARQ} is
equal to
$C(\frac{\gamma_{pp}}{1+\sum_{n=1}^{N}\varphi(a,n)\gamma_{s_np}})$
if in ARQ time $t$, action $a$ is selected.

The SNR region
${\Gamma}_{s_n,a,\Phi}(R_{s_1,a,\Phi},...,R_{s_N,a,\Phi})$,
$n\in\{1,...,N\}$, where $\phi(n)=K$, is the set of all $N-$tuples
of $SNRs$ $(\gamma_{s_1s_n},...,\gamma_{s_Ns_n})$, for which the
$SU_n$ message transmitted at rate $R_{s_n,a,{\Phi}}$ is
successfully decoded at $SU_{rxn}$ regardless of the decoding of
other SUs messages transmitted at rates $R_{s_m,a,\Phi}$, $\forall
m\neq n$. The SNR region
$\dot{\Gamma}_{s_n,a,\Phi}(R_p,R_{s_1,a,\Phi},...,R_{s_N,a,\Phi})$
is similarly defined for $\phi(n)=U$ and contains all SNR vectors
such that the $SU_n$ message transmitted at rate $R_{s_n,a,{\Phi}}$
is successfully decoded at $SU_{rxn}$ irrespective of the decoding
of other SUs and PU messages transmitted at rates $R_{s_m,a,\Phi}$
and $R_p$ respectively\footnote{Note that unlike in traditional
systems, where the decodability of a signal depends only on its own
rate, in the presence of Interference Cancelation it also depends on
the interferers' rates (see the examples in \eqref{E77} and
\eqref{E37}).}. Thus, the outage probability of the channel
$SU_{txn}\rightarrow SU_{rxn}$, $n\in\{1,...,N\}$ denoted by
$\rho_{s_n,a,\Phi}$ is computed as
\begin{align}\label{E31}
\rho_{s_n,a,\Phi=(\phi(1),...,\phi(n)=K,...,\phi(N))}=Pr\left((\gamma_{s_1s_n},...,\gamma_{s_Ns_n})\notin
{\Gamma}_{s_n,a,\Phi}(R_{s_1,a,\Phi},...,R_{s_N,a,\Phi})\right)
\end{align}
and
\begin{align}\label{E33}
\rho_{s_n,a,\Phi=(\phi(1),...,\phi(n)=U,...,\phi(N))}=Pr\left((\gamma_{psn},\gamma_{s_1s_n},...,\gamma_{s_Ns_n})\notin
{\dot{\Gamma}}_{s_n,a,\Phi}(R_p,R_{s_1,a,\Phi},...,R_{s_N,a,\Phi})\right).
\end{align}
As an example of how the SNR regions can be determined, we have:
\begin{align}\label{E77}
&{\dot{\Gamma}}_{s_1,1,\Phi}(R_p,R_{s_1,1,\Phi})\overset{\Delta}=\biggl\{(\gamma_{s_1s_1},\gamma_{ps_1}):
R_{s_1,1,\Phi}\leq C(\gamma_{s_1s_1}),\biggr.\nonumber\\
&\biggl.\;R_p\leq C(\gamma_{ps_1}), R_{s_1,1,\Phi}+R_p\leq
C(\gamma_{s_1s_1}+\gamma_{ps_1})\biggr\}\bigcup \nonumber\\
& \biggr\{(\gamma_{s_1s_1},\gamma_{ps_1}): R_p>C(\gamma_{ps_1}),
R_{s_1,1,\Phi}\leq
C(\frac{\gamma_{s_1s_1}}{1+\gamma_{ps_1}})\biggr\},\;\text{where
}\phi(1)=U,
\end{align}
\begin{align}\label{E37}
&{\Gamma}_{s_1,3,\Phi}(R_{s_1,3,\Phi},R_{s_2,3,\Phi})\overset{\Delta}=\biggl\{(\gamma_{s_1s_1},\gamma_{s_2s_1}):
R_{s_1,3,\Phi}\leq C(\gamma_{s_1s_1}),\biggr.\nonumber\\
&\biggl.\;R_{s_2,3,\Phi}\leq C(\gamma_{s_2s_1}),
R_{s_1,3,\Phi}+R_{s_2,3,\Phi}\leq
C(\gamma_{s_1s_1}+\gamma_{s_2s_1})\biggr\}\bigcup\nonumber\\
& \biggr\{(\gamma_{s_1s_1},\gamma_{s_2s_1}):
R_{s_2,3,\Phi}>C(\gamma_{s_2s_1}), R_{s_1,3,\Phi}\leq
C(\frac{\gamma_{s_1s_1}}{1+\gamma_{s_2s_1}})\biggr\},\;\text{where
}\phi(1)=K.
\end{align}
As observed, ${\dot{\Gamma}}_{s_1,1,\Phi}(R_p,R_{s_1,1,\Phi})$
depends on $R_p$ and $R_{s_1,1,\Phi}$, when $\phi(1)=U$. This is
because only $SU_1$ is allowed to access the channel when $a=1$ and
the $PU$ message is unknown at $SU_{rx1}$. It is also seen that
${\Gamma}_{s_1,3,\Phi}(R_{s_1,3,\Phi},R_{s_2,3,\Phi})$ depends on
$R_{s_1,3,\Phi}$ and $R_{s_2,3,\Phi}$ when $\phi(1)=K$. The reason
is that only $SU_1$ and $SU_2$ access the channel when $a=3$ and the
$PU$ message can be removed at the $SU_1$ receiver. All other SNR
regions can be similarly computed (full details for $N=2$ can be
found in \cite{R28}).


\subsection{Decentralized Scenario}\label{Rate-Dec}
In the decentralized case, each SU does not coordinate its access
action with the other SUs, and therefore there exist $N$ independent
binary access actions $a_n\in A_n=\{0,1\}$ $\forall
n\in\{1,...,N\}$. The access action $a$ in the decentralized case is
the combination of $N$ binary decisions (actions) $a_1,...,a_N$ and
may be derived as follows:
\begin{align}\label{E76}
a=Bin2Dec(a_N,a_{N-1},...,a_1)=\sum_{n=1}^N{ a_n 2^{n-1}},
\end{align}
where the function $Bin2Dec$ is binary to decimal conversion. Thus,
the rates and the outage probabilities at access action $a$ and PU
message knowledge state $\Phi$ defined in Section \ref{Rate-Cent}
can also be applied in the decentralized scenario.

\section{Centralized Optimal Access Policies for the SUs}\label{SIII}
The state of the $PU-SU_1-...-SU_N$ system may be modeled by a
Markov Process $s=(t,\Phi)$, where $t\in\{1,2,...,T\}$ is the
primary ARQ state and $\Phi$, the PU message knowledge state,
belongs to the set of $2^N$ possible combinations of PU message
knowledge states. The set of all states is indicated by
$\mathcal{S}$, and the number of states is equal to $2^N *(T-1)+1$.

The policy $\mu$ maps the state of the network $s$ to the
probability that the secondary users take access action
$a\in\{0,1,...,2^N-1\}$. The probability that action $a$ is selected
in state $s$ is denoted by $\mu(a,s)$. For example, with probability
$\mu(1,s)$, only $SU_{tx1}$ transmits and with probability
$\mu(0,s)=1-\sum_{i=1}^{2^N-1} \mu(i,s)$, they are all idle.

If access action $a\in\{1,,...,2^N-1\}$ is selected, the expected
throughput of $SU_n$, $n\in{1,...,N}$ in state $s=(t,\Phi)$ is
computed as
\begin{align}\label{E30}
{T}_{s_n,a,\Phi}=R_{s_n,a,\Phi}(1-\rho_{s_n,a,\Phi})
\end{align}
Since the model considered here is a stationary Markov chain, the
average long term SU sum throughput can be obtained as
\begin{align}\label{E2}
&\bar{T}_{su,c}(\mu)=\text{E}_{a,s=(t,\Phi)}\left[\sum_{n=1}^{N}{T}_{s_n,a,\Phi}\right]=\text{E}_{s=(t,\Phi)}\biggl[\sum_{a=1}^{2^N-1}\sum_{n=1}^{N}\mu(a,s)R_{s_n,a,\Phi}(1-\rho_{s_n,a,\Phi})\biggr],
\end{align}
where $\text{E}_{a,s}$ denotes the expectation with respect to $a$
and $s$. The outage probabilities $\rho_{s_n,a,\Phi}$ are given in
\eqref{E31} and \eqref{E33}.

The aim is to maximize the average long term sum throughput of the
SUs under the long term average PU throughput constraint, where the
average long term PU throughput is given by
$\bar{T}_{pu}=R_p\left(1-\sum_{a=0}^{2^N-1}\text{E}_{s=(t,\Phi)}\left[\mu(a,s)\right]\rho_{p,a}\right)$.
Using $\mu(0,s)=1-\sum_{a=1}^{2^N-1}\mu(a,s)$, the average long term
PU throughput $\bar{T}_{pu}$ is rewritten as follows:
\begin{align}\label{E1}
\bar{T}_{pu}&=R_p\left(1-\sum_{a=1}^{2^N-1}\text{E}_{s=(t,\Phi)}\left[\mu(a,s)\right]\rho_{p,a}\right)-R_p\left(\rho_{p,0}-\sum_{a=1}^{2^N-1}\text{E}_{s=(t,\Phi)}\left[\mu(a,s)\right]\rho_{p,0}\right)\nonumber
\\
&=T_{pu}^{I}-R_p\left(\sum_{a=1}^{2^N-1}\text{E}_{s=(t,\Phi)}\left[\mu(a,s)\right](\rho_{p,a}-\rho_{p,0})\right)\nonumber\\
&=T_{pu}^{I}-R_p\left(\text{E}_{a,s=(t,\Phi)}\left[{\rho_{p,a}-\rho_{p,0}}\right]\right),
\end{align}
where $T_{pu}^{I}=R_p(1-\rho_{p,0})$; and $\rho_{p,a}$,
$a\in\{0,...,2^N-1\}$ are given in \eqref{E34}.

Thus, if we request that $\bar{T}_{pu}\geq
T_{pu}^{I}(1-\epsilon_{PU})$, the PU throughput degradation
constraint is computed as follows
\begin{align}
T_{pu}^{I}-\bar{T}_{pu}=R_p\text{E}_{a,s=(t,\Phi)}&\left[{\rho_{p,a}-\rho_{p,0}}\right]\leq
R_p(1-\rho_{p,0})\epsilon_{PU}.\nonumber
\end{align}

Now we can formalize the optimization problem as follows:
\begin{problem}\label{P1}
\begin{align}\label{E22}
&\underset{\mu(a,s)}{\operatorname{maximize}}{\;\bar{T}_{su,c}(\mu)=\text{E}_{a,s=(t,\Phi)}\left[\sum_{n=1}^{N}{T}_{s_n,a,\Phi}\right]}\text{\;\;s.t.}
\end{align}
\begin{align}\label{E23}
&\text{E}_{a,s=(t,\Phi)}\left[{\rho_{p,a}-\rho_{p,0}}\right]\leq
(1-\rho_{p,0})\epsilon_{PU}\triangleq\epsilon_{\omega},
\end{align}
where $\mu(a,s)$ is the probability that access action $a$ is
selected in state $s$.
\end{problem}
The constraint \eqref{E23} is referred to as the normalized PU
throughput degradation constraint.

To give a solution to Problem \ref{P1}, we provide the following
definition, which identifies the boundary between low and high
access rate regimes.
\begin{definition} \label{D1}
Let $\acute{\mu}_{init}=\{\mu_{1,init},...,\mu_{2^{N}-1,init}\}$ be
defined as follows:
\begin{align}\label{E7}
\acute{\mu}_{init}=
\begin{cases} \mathbf{U}_m& \forall
s\in\mathcal{S_K}=\left\{(t,(K,...,K)):t\in\{1,2,...,T\}\right\}\\
0& \forall s\notin\mathcal{S_K},
\end{cases}
\end{align}
where
\begin{align}\label{E35}
m=\underset{a\in\{1,...,2^N-1\}}{\arg\max}v_a
\end{align}
\begin{align}\label{E36}
v_a=Dec2Bin_N(a).({T}_{s_1,a,(K,...,K)},...,{T}_{s_N,a,(K,...,K)})\operatorname{min}(\frac{\epsilon_{\omega}}{\rho_{p,a}-\rho_{p,0}},1)
\end{align}
and $A.B$ is the inner product of two vectors A and B; and
${T}_{s_n,a,(K,...,K)}$ is given in \eqref{E30}. Thus, according to
\eqref{E35},  action $m\in\{1,...,2^N-1\}$ is selected if
$s\in\mathcal{S_K}$, otherwise action $0$ is selected. Note that
${\mu}_{init}=\{\mu_{0,init}\}\bigcup\acute{\mu}_{init}$ is a random
access policy, where
$\mu_{0,init}=1-\sum_{a=1}^{2^N-1}{\mu_{a,init}}$.

For access policy ${\mu}_{init}$, we compute the normalized PU
throughput degradation constraint in \eqref{E23} and refer to it as
$\omega_{init}$. Hence, replacing \eqref{E7} in \eqref{E23} and then
computing the expectation with respect to $a$ and $s$,
$\omega_{init}$ can be obtained as follows:
\begin{align}\label{E8}
&\omega_{init}=({\rho_{p,m}-\rho_{p,0}})\sum_{t=1}^{T}{\pi(t,(K,...,K))},
\end{align}
where $m$ is given in \eqref{E35} and $\pi(t,(K,...,K))$ is the
steady-state probability of being in state $s=(t,(K,...,K))$.
\end{definition}

 In the sequel, we derive an upper bound to the
average long term sum throughput of SUs, and characterize the low SU
access rate regime $\epsilon_{\omega}\leq\omega_{init}$ and high SU
access rate regime $\epsilon_{\omega}>\omega_{init}$.

\subsection{Upper Bound to the Average Long Term SU Sum
Throughput in Centralized Access Policy Design}\label{UPPER} An
upper bound to the average long term SU sum throughput is achieved
when the receivers are assumed to know the PU message, so that they
can always cancel the PU interference. Since each SU always knows
the PU message, as in \cite{R9} there exists an optimal access
policy which is independent of the ARQ state, and therefore is the
same in each slot. We refer to this policy as
$\mu=\{\mu_0,\mu_1,...,\mu_{2{N}-1}\}$. Thus, noting that
$\sum_{a=1}^{2^N-1}\mu_a\leq1$ and $0\leq\mu_{a}$, Problem \ref{P1}
may be rewritten as follows:
\begin{problem}\label{P2}
\begin{align}\label{E38}
&\underset{\mu_1,...,\mu_{2^{N}-1}}{\operatorname{max}}\bar{T}_{su,c}(\mu)=\sum_{a=1}^{2^N-1}\mu_a{Dec2Bin_N(a).({T}_{s_1,a,(K,...,K)},...,{T}_{s_N,a,(K,...,K)})},\;\;s.t.
\end{align}
\begin{align}\label{E39}
         &\sum_{a=1}^{2^N-1}\mu_a({\rho_{p,a}-\rho_{p,0}})\leq\epsilon_{\omega},
         \;\;\sum_{1}^{2^N-1}\mu_a\leq1,
         \end{align}
\end{problem}
where $0\leq\mu_a$. Proposition \ref{PR2} below provides a solution
to Problem \ref{P2}.
\begin{proposition}\label{PR2}
An access policy to achieve the upper bound is given by
$\mu^{u}=\{\mu^u_0,\mu^u_1,...,\mu^u_{2^N-1}\}=\{\mu^{u}_0\}\bigcup\acute{\mu}^{u}$,
where\footnote{Please note that the ``min'' operation in \eqref{E91}
and \eqref{E28} (which was erroneously not included in \cite{R22})
is needed to ensure that $\mu^u$ is a valid probability distribution
when $\frac{\epsilon_w}{\rho_{p,i}-\rho_{p,0}}>1$.}\vspace{-4pt}
\begin{align}\label{E91}
&\acute{\mu}^{u}=\operatorname{min}(\frac{\epsilon_{\omega}}{\rho_{p,m}-\rho_{p,0}},1)\mathbf{U}_m
\end{align}
and the upper bound to the average long term SU sum throughput is
obtained as
\begin{align}\label{E28}
&\bar{T}_{su,c}^{u}=\operatorname{min}(\frac{\epsilon_{\omega}}{\rho_{p,m}-\rho_{p,0}},1){Dec2Bin_N(m).({T}_{s_1,m,(K,...,K)},...,{T}_{s_N,m,(K,...,K)})}
\end{align}
where $m$ is defined in \eqref{E35} and the other parameters are
given in Sections \ref{system model} and \ref{Rate}.
\end{proposition}
\begin{proof}
Using Lagrange multipliers $\lambda_1$ and $\lambda_2$, the
Lagrangian for Problem \ref{P2} is
\begin{align}
L=&\sum_{a=1}^{2^N-1}\mu_a{Dec2Bin_N(a).({T}_{s_1,a,(K,...,K)},...,{T}_{s_N,a,(K,...,K)})}-\lambda_1\biggl(\sum_{a=1}^{2^N-1}\mu_a({\rho_{p,a}-\rho_{p,0}})-\epsilon_{\omega}\biggr)-\nonumber\\
&\lambda_2(\sum_{a=1}^{2^N-1}\mu_a-1)
\end{align}\vspace{-4pt}
and then the Kuhn-Tucker conditions are as follows:
\begin{align}\label{E92}
&\frac{\partial L}{\partial\mu_i}\leq 0,\;\;\;\mu_i\geq0,
\;\;\;\mu_i\frac{\partial
L}{\partial\mu_i}=0\;\;\;\;\;i\in\{1,...,2^{N}-1\}\\\label{E93}
&\sum_{a=1}^{2^N-1}\mu_a({\rho_{p,a}-\rho_{p,0}})-\epsilon_{\omega}\leq
         0,\;\;\;\lambda_1\geq 0,\;\;\;\lambda_1\biggl(\sum_{a=1}^{2^N-1}\mu_a({\rho_{p,a}-\rho_{p,0}})-\epsilon_{\omega}\biggr)=0\\\label{E94}
&\sum_{a=1}^{2^N-1}\mu_a-1\leq 0,\;\;\;\lambda_2\geq
0,\;\;\;\lambda_2(\sum_{a=1}^{2^N-1}\mu_a-1)=0.
\end{align}
To solve the problem, we need to consider different situations for
the various inequalities. The complete proof is given in Appendix
\ref{A1}.
\end{proof}
Thus, if $m=1$ is the answer to \eqref{E35}, then only $SU_1$ can
access the channel while satisfying the PU throughput degradation
constraint. Thus, $v_1$ is proportional to the ratio of the $SU_1$
throughput over the relative PU throughput. Relative PU throughput
indicates the amount of reduction in the PU throughput if only
$SU_1$ transmits with respect to that when no SU transmits. The
result for the other selected $m$ can be interpreted in a similar
way.

\subsection{Low SU Access Rates Regime in Centralized Access Policy Design}
Now we consider the low SU access rate regime
$\epsilon_{\omega}\leq\omega_{init}$, where $\epsilon_{\omega}$ is
defined in \eqref{E23}. Proposition \ref{PR1} below characterizes
the optimum access policy for this access rate regime.

\begin{proposition}\label{PR1}
In the low SU access rate regime
$\epsilon_{\omega}\leq\omega_{init}$, the optimal access policy is
given by
\begin{align}\label{E6}
\mu^{*}=\{\mu^*_0,\mu^*_1,...,\mu^*_{2^N-1}\}=\{\mu^{*}_0\}\bigcup\acute{\mu}^{*},
\end{align}
where
\begin{align}\label{E41}
&\acute{\mu}^{*}=\begin{cases}(\frac{\epsilon_{\omega}}{\omega_{init}})\mathbf{U}_m&
\forall
s\in\mathcal{S_K}=\left\{(t,(K,...,K)):t\in\{1,2,...,T\}\right\}\\
0& \forall s\notin\mathcal{S_K},
\end{cases}
\end{align}
and
\begin{align}\label{E42}
&\bar{T}_{su,c}^{u}=(\frac{\epsilon_{\omega}}{\rho_{p,m}-\rho_{p,0}}){Dec2Bin_N(m).({T}_{s_1,m,(K,...,K)},...,{T}_{s_N,m,(K,...,K)})}
\end{align}
where $m$ is defined in \eqref{E35} and the other parameters are
given in Sections \ref{system model} and \ref{Rate}.
\end{proposition}
\begin{proof}
With $\mu_{init}$ in \eqref{E7} (Definition \ref{D1}), the
constraint \eqref{E23} is equal to $\omega_{init}$ as given in
\eqref{E8}. However, for the low SU access rate regime,
$\epsilon_{\omega}$ is less than or equal to $\omega_{init}$. To
meet this stricter constraint, we can scale the access policy
$\mu_{init}$ in \eqref{E7} by
$\frac{\epsilon_{\omega}}{\omega_{init}}$ such that \eqref{E23} is
satisfied with equality. Therefore, $\mu^*$ in \eqref{E6} satisfies
the constraint. Replacing $\mu^*$ in \eqref{E22} we obtain
\begin{align}\label{E98}
\bar{T}_{su,c}(\mu)=\frac{\epsilon_{\omega}}{\omega_{init}}{Dec2Bin_N(m).({T}_{s_1,m,(K,...,K)},...,{T}_{s_N,m,(K,...,K)})}\sum_{t=1}^{T}{\pi(t,(K,...,K))}
\end{align}
Thus, substituting $\omega_{init}$ given in \eqref{E8} results in
the SU sum throughput as given in \eqref{E42}. Since the SU sum
throughput \eqref{E42} is equal to the upper bound \eqref{E28} in
the low SU access rate regime $\epsilon_{\omega}\leq\omega_{init}$,
the proposed access policy \eqref{E6} is optimal. Note that in the
low SU access rate regime since
$\epsilon_{\omega}\leq\omega_{init}$, we have
\begin{align}\label{E99}
\frac{\epsilon_{\omega}}{\rho_{p,m}-\rho_{p,0}}\leq1,
\end{align}
where $m$ is defined in \eqref{E35}.
\end{proof}
Proposition \ref{PR1} provides the conditions in which the SUs can
access the channel in the low SU access rate regime. As observed,
the SUs are not allowed to transmit if even one of the SU receivers
does not know the PU message.

\subsection{High SU Access Rates Regime in Centralized Access Policy Design}\label{HA}
In Problem \ref{P1}, we are looking for an optimum policy for the
CMDP problem. Therefore, for high SU access rate regime, we employ
the equivalent LP formulation corresponding to CMDP, e.g., see
\cite{R10}, \cite{R13}. To provide the equivalent LP, we need the
transition probability matrix of the Markov process denoted by $P$,
where $P_{s\acute{s},a}$ is the probability of moving from state $s$
to $\acute{s}$ if access action $a$ is chosen. Note that the only
allowable state with ARQ time $t=1$ is state $s=(1,(U,...,U))$ and
the process always restarts from $s=(1,(U,...,U))$ when the maximum
number of PU retransmissions or the successful decoding of the PU
messages occurs. To obtain the transition probability matrix
$P_{{s}\acute{{s}},a}$, we need to compute the transition
probability matrix of the PU Markov model $Q_{t\acute{t},a}$ as
given in \eqref{E101}, which is the probability that the primary ARQ
state $t$ is transferred to $\acute{t}$ if access action $a$ is
selected.
\begin{align}\label{E101}
&Q_{t\acute{t},a}=\begin{cases} 1&\text{if }\acute{t}=1,t=
T\\
1-\rho_{p,a}&\text{if }
\acute{t}=1,t\neq T\\
 \rho_{p,a}&\text{if }\acute{t}=t+1,t\neq T \\
0&\text{otherwise}.
\end{cases}
\end{align}

Thus, $P_{s\acute{s},a}=P_{({t},{\Phi})(\acute{t},\acute{\Phi}),a}$
is given by
\begin{align}
P_{({t},{\Phi})(\acute{t},\acute{\Phi}),a}=Q_{t\acute{t},a}Pr(\acute{\Phi}|\Phi,a),
\end{align}
where $Pr(\acute{\Phi}|\Phi,a)$, the probability that the PU message
knowledge state $\Phi$ is changed to state $\acute{\Phi}$ given
action $a$, is obtained as follows:
\begin{align}\label{E102}
Pr(\acute{\Phi}|\Phi,a)=\prod_{n=1}^N F_n({\Phi},a),
\end{align}
where
\begin{align}
F_n({\Phi},a)=\begin{cases}\rho_{ps_n,a,\Phi}&\text{if } \phi(n)=U
\text{ and }
\acute{\phi}(n)=U\\
1-\rho_{ps_n,a,\Phi}&\text{if } \phi(n)=U \text{ and }
\acute{\phi}(n)=K\\
1&\text{if } \phi(n)=K \text{ and }
\acute{\phi}(n)=K\\
0&\text{if } \phi(n)=K \text{ and } \acute{\phi}(n)=U
\end{cases}
\end{align}
and $\rho_{ps_n,a,\Phi}$ is the probability that $SU_{rxn}$ is not
able to decode the PU message in PU message knowledge state $\Phi$
if access action $a$ is selected.

For any unichain Constrained Markov Decision Process, there exists
an equivalent LP formulation, where an MDP is unichain if it
contains a single recurrent class plus a (perhaps empty) set of
transient states \cite{R29}. Since the transition probability of
moving from every state to state $s=(1,\{U,...,U\})$ is not zero,
our CMDP model is unichain. Thus, the following problem formalizes
the equivalent LP for Problem \ref{P1} \cite{R10}:
\begin{problem}\label{P3}
\begin{align}
&\underset{x}{\operatorname{maximize}}
\sum_{s\in\mathcal{S}}\sum_{a\in A}{\sum_{n=1}^N {T}_{s_n,a,\Phi}x(s,a)}\text{\;\;s.t.}\\
&\sum_{s\in\mathcal{S}}\sum_{a\in A}{(\rho_{p,a}-\rho_{p,0})x(s,a)}\leq\epsilon_{\omega}\\
 &\sum_{a\in A}x(\acute{s},a)-\sum_{s\in\mathcal{S}}\sum_{a\in
A}{P_{s\acute{s},a}\;x(s,a)}=0
\;\;\;\forall\acute{s}\in\mathcal{S}\\
&\sum_{s\in\mathcal{S}}\sum_{a\in A}{x(s,a)}=1 \\
&x(s,a)\geq 0\;\;\;\forall s\in\mathcal{S},\;a\in A.
\end{align}
\end{problem}
Note that since the number of actions and states are respectively
equal to $2^N$ and $2^N *(T-1)+1$, the computational complexity of
the LP approach is the order of $2^{2N}$. The relationship between
the optimal solution of Problem \ref{P3} and the solution to the
considered Problem \ref{P1} is obtained as follows \cite{R10}:
\begin{align}
\mu(a,s)=\begin{cases}\frac{x(s,a)}{\sum_{\acute{a}\in
A}{x(s,\acute{a})}}&\text{if}\sum_{\acute{a}\in A}{x(s,\acute{a})}>0\\
\text{arbitrary}&\text{otherwise}.
\end{cases}
\end{align}

All cases of practical interest considered in this paper correspond
to a unichain CMDP. For the equivalent linear problem corresponding
to the general case of a multichain CMDP, the reader is referred to
\cite{R13}.

\section{Decentralized Access Policies for SUs in MMDP Model}\label{S_DeC}
In this section, we assume that there is no central unit to control
the access policy of the SU transmitters. Therefore, each SU has to
control its own access policy independently. We also assume that the
PU message knowledge state of each SU receiver is known to all SU
transmitters (e.g., the $SU_{rxm}$ sends back its PU message
knowledge state on an error free feedback channel, which is heard by
all SU transmitters). Hence, the state $s$ defined in Section
\ref{SIII} is known to all transmitters. However, since there is no
central unit, there is no coordination among the SUs, and $SU_m$
does not know the action selected by $SU_n$, $n\neq m$. Thus, each
secondary user knows the state of the MDP but not the action
selected by the other users. In this case, the $PU-SU_1-...-SU_N$
system may be modeled by an Multi-agent Markov Decision Process
$s=(t,\Phi)$ \cite{R23}, where $t$ and $\Phi$ are defined in Section
\ref{SIII}. The set of all states is indicated by $\mathcal{S}$. In
contrast to the centralized scenario, we have $N$ policies $\mu_n$,
$n=\{1,...,N\}$, which map the state of the network $s$ to the
probabilities that each secondary user $n$ takes access action
$a_n\in A_n=\{0,1\}$. The probability that action $a_n$ is selected
by $SU_n$ in state $s$ is denoted by $\mu_n(a_n,s)$, where $a_n=0$
if $SU_{txn}$ does not transmit, and $a_n=1$ otherwise ($SU_{txn}$
transmits). We use the notation $\mu=(\mu_1,...,\mu_N)$ for the
access policy of the system in the decentralized case.  As denoted,
the objective is to maximize the average long term sum throughput of
the SUs under the long term average PU throughput constraint as
formalized in Problem \ref{P5}, where all throughputs are influenced
by the actions selected by the $N$ users.
\begin{problem}\label{P5}
\begin{align}\label{E5}
&\underset{\mu_1(a_1,s),...,\mu_N(a_N,s)}{\operatorname{maximize}}{\;\bar{T}_{su,d}(\mu_1,...,\mu_N)=
\text{E}_{a,s=(t,\Phi)}\left[{T}_{s_1,a,\Phi}+...,+{T}_{s_N,a,\Phi}\right]}\text{\;\;s.t.}
\end{align}
\begin{align}\label{E16}
&D(\mu_1,...,\mu_N)=\text{E}_{a,s=(t,\Phi)}\left[{\rho_{p,a}-\rho_{p,0}}\right]\leq
\epsilon_{\omega},
\end{align}
where $a=Bin2Dec(a_N,a_{N-1},...,a_1)$, $\epsilon_{\omega}$ is
defined in Section \ref{SIII}; and $\mu_n(a_n,s)$ is the probability
that access action $a_n$ is selected at transmitter $SU_n$, given
system state $s$.
\end{problem}
Since the access policy designed in Section \ref{SIII} is a
randomized policy \cite{R10}, in general we can not find an access
policy for each SU from the proposed centralized access policy. For
example, assume $N=2$ and the centralized optimum policy
$\mu=[0.3,\; 0,\; 0,\; 0.7]$, which cannot be implemented in a
distributed way. This is because that we can find the two
probabilities $\nu_1=\mu_1(0,s)$ and $\nu_2=\mu_2(0,s)$ by solving
the two equations $\nu_1 \nu_2 = 0.3$ and $(1-\nu_1)(1-\nu_2)=0.7$,
but the solution would be incompatible with
$\nu_1(1-\nu_2)=\nu_2(1-\nu_1)=0$. This is actually a result of the
fact that in the centralized solution we pick a probability
distribution over $2^N$ values, which has $2^N-1$ degrees of
freedom, whereas in the decentralized scenario we pick $N$ binary
distributions, with only $N$ degrees of freedom, and therefore there
always exist centralized distributions that cannot be obtained by
combining $N$ binary distributions for any $N>1$.

In the sequel, a scheme based on Nash Equilibrium is proposed, which
finds the local optimum policies by converting the CMMDP to a CMDP
\cite{R19}, \cite{R20}.

\subsection{Decentralized Access policy Design Using Nash Equilibrium}\label{AS-E}
We employ Nash Equilibrium, in which no user has an interest in
unilaterally changing its policy. In fact, $SU_n$ transmitter
designs its optimal policy by assuming fixed policies for the other
SUs. This procedure for different SUs continues until there is no
benefit in employing more iterations. Assuming fixed policies
$\mu_n$ for $SU_n$, the problem for $SU_m$, $m\neq n$ can be
considered as a CMDP, referred to as $CMDP_m$. The state space of
the new model is the same as the system state $\mathcal{S}$. In
fact, since the system state $s$ is known for all users, the state
of $CMDP_m$ is $s=(t,(\phi(1),...,\phi(N)))$. $SU_{txm}$ chooses
action $a_m$ from the set $A_m=\{0,1\}$, where for $a_m=1$ and
$a_m=0$ the $SU_m$ does or does not transmit respectively. Problem
\ref{P7} below formalizes the new optimization problem for $SU_m$
assuming fixed stationary policies for all $SU_n$, $n\neq m$.
\begin{problem}\label{P7}
\begin{align}\label{E73}
&\underset{\mu_m(a_m,s)}{\operatorname{maximize}}{\;\text{E}_{a_m,s=(t,\Phi)}\left[\sum_{a_n,\forall
n\neq m}({T}_{s_1,a,\Phi}+...+{T}_{s_N,a,\Phi})\prod_{n=1,n\neq
m}^N\mu_n(a_n,s)\right]}\text{\;\;s.t.}
\end{align}
\begin{align}\label{E16}
&\text{E}_{a_m,s=(t,\Phi)}\left[\sum_{a_n,\forall n\neq
m}({\rho_{p,a}-\rho_{p,0}})\prod_{n=1,n\neq
m}^N\mu_n(a_n,s)\right]\leq \epsilon_{\omega},
\end{align}
where $\epsilon_{\omega}$ is defined in Section \ref{SIII}; and
$\mu_m(a_m,s)$ is the probability that access action $a_m$ is
selected in state $s$ by the $SU_m$ transmitter.
\end{problem}
Assume a fixed stationary policy $\mu_n$ for $SU_n$, $\forall
n\in\{1,...,N\}$, $n\neq m$. The problem for $SU_m$ is a CMDP
characterized by tuple $(s,\acute{P}^m,\acute{r}^m,\acute{d}^m)$,
where
\begin{align}
\acute{P}^m_{s\acute{s},a_m}=\sum_{a_n,\forall n\neq
m}P_{s\acute{s},a}\prod_{n=1,n\neq m}^N\mu_n({a_n},s),
\end{align}
\begin{align}
\acute{r}^m_{s,a_m}=\sum_{a_n,\forall n\neq
m}({T}_{s_1,a,\Phi}+...+{T}_{s_N,a,\Phi})\prod_{n=1,n\neq
m}^N\mu_n({a_n},s),
\end{align}
\begin{align}
\acute{d}^m_{s,a_m}=\sum_{a_n,\forall n\neq
m}({\rho_{p,a}-\rho_{p,0}})\prod_{n=1,n\neq m}^N\mu_n({a_n},s).
\end{align}
 $\acute{P}^m$, $\acute{r}^m$ and $\acute{d}^m$, respectively are
the transition matrix probability, the instantaneous reward function
and the instantaneous cost function in the new model and
 $P_{s\acute{s},a}$ is the transition probability of the
system.

As explained in Section \ref{HA}, there is an equivalent LP
formulation for any unichain CMDP, and the LP formulation
corresponding to $CMDP_m$ described in Problem \ref{P7}  is given by
\begin{problem}\label{P8}
\begin{align}
&\underset{x^m}{\operatorname{maximize}}
\sum_{s\in\mathcal{S}}\sum_{a_m\in A_m}{\acute{r}^m_{s,a_m}x^m(s,a_m)}\text{\;\;s.t.}\nonumber\\
&\sum_{s\in\mathcal{S}}\sum_{a_m\in A_m}{\acute{d}^m_{s,a_m}x^m(s,a_m)}\leq\epsilon_{\omega}\nonumber\\
 &\sum_{a_m\in A_m}x^m(\acute{s},a_m)-\sum_{s\in\mathcal{S}}\sum_{a_m\in
A_m}{\acute{P}^m_{s\acute{s},a_m}\;x^m(s,a_m)}=0
\;\;\;\forall\acute{s}\in\mathcal{S}\nonumber\\
&\sum_{s\in\mathcal{S}}\sum_{a_m\in A_m}{x^m(s,a_m)}=1\nonumber \\
&x^m(s,a_m)\geq 0\;\;\;\forall s\in\mathcal{S},\;a_m\in A_m.
\end{align}
\end{problem}
The relationship between the optimal solution of LP Problem \ref{P8}
and the solution to the considered Problem \ref{P7} is also obtained
as follows
\begin{align}\label{E74}
\mu_m({a_m},s)=\begin{cases}\frac{x^m(s,a_m)}{\sum_{\acute{a_m}\in
A_m}{x^m(s,\acute{a_m})}}&\text{if}\sum_{\acute{a_m}\in A_m}{x^m(s,\acute{a_m})}>0\\
\text{arbitrary}&\text{otherwise}.
\end{cases}
\end{align}

As denoted, $SU_m$ computes the optimum policy as given in
\eqref{E74} by considering fixed policies for other $SUs$. By
changing $m\in\{1,...,N\}$, this procedure iteratively continues
until an equilibrium is achieved. (We prove later in Proposition
\ref{PR9} that an equilibrium point is always achieved.) Algorithm 1
below describes the local optimal solution to Problem \ref{P5} based
on Nash Equilibrium. The obtained access policies are local optimum
solutions. We have to restart Algorithm 1 for several random
initiations and see whether the resulting SU sum throughput is
higher.
\begin{algorithm*}\label{AL1}
\caption{Local Optimum Policy using Nash Equilibrium}
\begin{enumerate}
\item Choose initial stochastic policies $\mu_n$ $\forall n\in\{1,...,N\}$ and
select $m=1$, $l=1$ and $\mu^1=(\mu_1 ,...,\mu_N)$.
\item To compute optimum policy
$\mu_m$ for $SU_m$, obtain the solution to Problem \ref{P7} as given
in \eqref{E74} for given $\mu_n,$ $\forall n\neq m$,
\item Select $l=l+1$, $m=m+1$ and $\mu^l=(\mu_1 ,...,\mu_N)$. If $m=N+1$, then
$m=1$.
\item If $\mu^l=\mu^{l-1}$, then go step 5. Else go step 2.
\item  $\mu_n,\;\forall n\in\{1,...,N\}$ are the local optimum
solution to original DEC-MMDP Problem \ref{P5}.
\end{enumerate}
\end{algorithm*}
We have the two following propositions related to Nash Equilibrium.

\begin{proposition}\label{PR8}
Optimum access policies $\mu_n^*,\;\forall n\in\{1,...,N\}$ solution
to Problem \ref{P5} are a fixed point or an equilibrium point.
\end{proposition}
\begin{proof}
If $\mu_n^*,\;\forall n\in\{1,...,N\}$ are the optimum solutions to
problem \ref{P5}, then
\begin{align}\label{E75}
\bar{T}_{su,d}(\mu^*_1,...,\mu^*_N)\geq
\bar{T}_{su,d}(\mu_1,...,\mu_N),
\end{align}
where $(\mu_1,...,\mu_N)$ belongs to the set of all feasible
solutions (i.e., the set of polices that satisfy the constraint in
Problem \ref{P5}) and $\bar{T}_{su,d}(\mu^*_1,...,\mu^*_N)$ is given
in Problem \ref{P5}. Now suppose that the policy for $SU_n\;\forall
n\neq 1$ is fixed to $\mu_n^*$. Note that $\bar{T}_{su,d}(\mu_{1})$
in Problem \ref{P7} is equal to
$\bar{T}_{su,d}(\mu_1,\mu^*_2,...,\mu^*_N)$ in Problem \ref{P5}.
Thus, noting \eqref{E75}, we have
\begin{align}
\bar{T}_{su,d}(\mu_{1})\leq
\bar{T}_{su,d}(\mu^*_1,\mu^*_2,...,\mu^*_N)
\end{align}
and if $\mu_1$ is equal to $\mu^*_1$, equality occurs. Thus, point
$(\mu^*_1,...,\mu^*_N)$ is a fixed point. In other words, this fixed
point is an equilibrium where no user can get any more benefit in SU
sum throughput by more iterations.
\end{proof}

\begin{proposition}\label{PR9}
The SU sum throughput obtained by solving Algorithm 1 improves as
the iteration index $l$ increases and furthermore, the iterative
procedure based on Algorithm 1 converges to a fixed point.
\end{proposition}

\begin{proof}
Suppose $\mu^{l}_n\;\forall n\in\{1,...,N\} $  are the resulting
policies in iteration $l$ of the algorithm and the resulting SU sum
throughput is given by $\bar{T}_{su,d}(\mu^{l}_1,...,\mu^{l}_N)$.
Now we consider $\mu^{l}_n\;n\neq 1$ to be fixed and improve
$\mu^{l}_1$ to $\mu^{l+1}_1$ according to the algorithm. Therefore,
$\mu^{l+1}_1$ is the optimum solution to Problem \ref{P7} and we
have
\begin{align}
\bar{T}_{su,d}(\mu^{l+1}_1)\geq \bar{T}_{su,d}(\mu^{l}_1)
\end{align}
or equivalently
\begin{align}
\bar{T}_{su,d}(\mu^{l+1}_1,\mu^{l}_2,...,\mu^{l}_N)\geq
\bar{T}_{su,d}(\mu^{l}_1,\mu^{l}_2,...,\mu^{l}_N).
\end{align}
Since $\bar{T}_{su,d}(\mu^{l}_1,...,\mu^{l}_N)$ and
$\bar{T}_{su,d}(\mu^{l+1}_1,\mu^{l}_2,...,\mu^l_N)$ are the SU sum
throughput respectively in iterations $l$ and $l+1$, it is observed
that the SU sum throughput can not decrease as the algorithm
proceeds. The same approach could be seen when the policy for
$SU_n,\;n\neq1$ improves and the policies of the other SUs are
constant. This shows that the SU sum throughput is an increasing
function with respect to $l$. Since the performance is bounded by
that of the centralized access policy design, it is proved that the
proposed algorithm converges.
\end{proof}
Propositions \ref{PR8} and  \ref{PR9} prove that the optimum
solution to the decentralized access policy design is an equilibrium
point and the decentralized access policy design based on Algorithm
1 converges to a fixed point.

\section{Numerical Results}\label{SIV}
For numerical evaluations we consider a CRN with $N$ SUs,
$N\in\{1,2,3\}$, and Rayleigh fading channels. Thus, the SNR
$\gamma_{x}$ is an exponentially distributed random variable with
mean $\bar{\gamma}_{x}$, where $x\in\{pp,ps_n,s_ns_m,s_np\}$,
$n,m\in\{1,...,N\}$. We consider the following parameters throughout
the paper, unless otherwise mentioned. Following \cite{R9}, we
consider the average SNRs $\bar{\gamma}_{pp}=10$,
$\bar{\gamma}_{s_ns_n}=5$, $\bar{\gamma}_{ps_m}=5$,
$\bar{\gamma}_{s_np}=2$, and $\bar{\gamma}_{s_ns_m}=3$, $m\neq n$.
The ARQ deadline is $T=5$. The PU rate $R_p$ is selected such that
the PU throughput is maximized when all SUs are idle, i.e.,
$R_p=\operatorname{argmax}_{R}{T_{pu}^{I}(R)}$. Thus, we set
$R_p=2.52$ and $T_{pu}^{I}=1.57$. The PU throughput constraint is
set to $(1-\epsilon_{PU}) T_{pu}^{I}$, where $\epsilon_{PU}=0.2$. In
the centralized case the rates $R^*_{s_n,a,\Phi}$,
$n\in\{1,...,N\}$, are computed as
$(R^*_{s_1,a,\Phi},...,R^*_{s_N,a,\Phi})=\operatorname{argmax}_{R_{s_1,a,\Phi},...,R_{s_N,a,\Phi}}{T_{s_1,a,\Phi}+...+T_{s_N,a,\Phi}}$
so as to maximize the SU sum throughput, where $T_{s_n,a,\Phi}$,
$n\in\{1,...,N\}$ is a function of $R_P$ (only if the PU message
knowledge state is unknown for receiver $SU_n$) and of all
$R_{s_m,a,\Phi}$, $\forall m\in\{1,...,N\}$. In the decentralized
case, the rate $R_{s_n,a,\Phi}$ is selected so as to maximize
$T_{s_n,a,\Phi}$, irrespective of the other SU transmissions.

We remark that the SU access policies are randomized, in the sense
that, for a given system state, different channel access outcomes
are possible with different probabilities. In the centralized case,
the policy is given by the joint distribution of the channel access
actions by all $N$ SUs, whereas in the decentralized case each SU
makes its own randomized binary decision about whether or not to
access the channel.

The scheme ``Forward Interference Cancelation'' discussed here is
called ``FIC''. The centralized and decentralized access policy
designs are respectively referred to as ``FIC Decentralized'' and
``FIC Centralized''. For the centralized policy design, the
performance bound described in Section \ref{UPPER} is referred to as
``PM already Known''. To validate the SU sum throughput obtained by
Problem 3, we use access policies proposed by "FIC Centralized" in a
Monte-Carlo simulation, compute the SU sum throughput and PU
throughput degradation and refer to it as ``FIC
Centralized-Monte-Carlo''. In addition, we also consider the
scenario without using FIC, referred to as ``No FIC'' in the
centralized access policy design. Note that ``FIC: One SU'' denotes
the case that only one SU exists in the CRN and its receiver applies
FIC.


The SU sum throughput with respect to the PU throughput by varying
the value of $\epsilon_{PU}$ for $N=1,\;2,\;3$ is depicted in Fig.
\ref{F3}. As can be observed from Fig. \ref{F3}, ``FIC
Centralized-Monte-Carlo'' matches the SU sum throughput obtained by
the solution to Problem 3. It is also obvious that as the PU
throughput $T_{pu}^{I}(1-\epsilon_{PU})=1.57(1-\epsilon_{PU})$
increases, the average sum throughput of SUs decreases. PU
throughputs greater than $1.286$ and $1.224$ ($\epsilon_{PU}<0.22$
and $\epsilon_{PU}<0.18$) correspond to the low SU access rate
regime respectively for centralized and decentralized scenarios with
two SUs. The FIC performance is the same as that of the upper bound
(``PM already Known'' scheme) for the low SU access rate regime. As
can be observed, each CRN scenario provides a constant SU sum
throughput for a loose enough constraint on the PU throughput. There
is also a performance loss with applying the decentralized approach
with respect to the centralized one in CRN with either $N=2$ or
$N=3$ SUs, especially for a loose PU throughput constraint. Our
simulation results show that this loss in the decentralized scenario
is because the assigned rate to each SU does not account for the
decision made by the other SUs, whereas in the centralized case the
rates are jointly assigned. In fact, when the rates assigned to the
SUs in the decentralized case are the same as those in the
centralized case, our proposed decentralized design has the same
performance as the centralized design. It can be seen that the
decentralized scenario with $N=3$ provides a performance similar to
 $N=2$ even for a loose PU throughput constraint and this is
because the SUs interfere more with each other when the rate at each
SU is assigned irrespective of the other SUs. Thus, increasing the
number of SUs generates more interference at the SU receivers and
requires the SUs to reduce their access to the channel. The results
also reveal that the trend of the tradeoff curve between PU and SU
sum throughput is the same for all FIC schemes regardless of $N$,
and the slope of the tradeoff curves after departing from the ``PU
always known'' curve is the same in all cases including ``No FIC'',
where the difference among the various cases is the value on which
the various curves settle in the loose PU throughput constraints.
Thus, from the results of Fig. \ref{F3}, it can be concluded that,
despite the obvious quantitative differences (SU sum throughput is
higher when more SUs are present and the PU throughput constraint is
loose), the trends of all curves are very similar. For this reason,
in the rest of this section, in order to keep the plots more
readable, we will focus on the simpler case $N=2$, with the
understanding that for $N=3$ we will have curves with similar
behaviors and slightly better throughput.

\begin{figure}[t!]
\centering
\includegraphics[width=3.5 in]{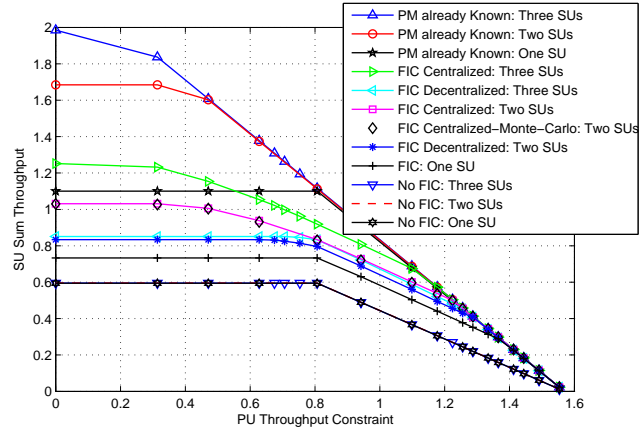}\caption{Average sum throughput of SUs with respect to PU throughput constraint $(1-\epsilon_{PU}) T_{pu}^{I}$. $\bar{\gamma}_{ps_n}=5$, $\bar{\gamma}_{s_np}=2$,
 $\bar{\gamma}_{pp}=10$, $\bar{\gamma}_{s_ns_n}=5$ and $\bar{\gamma}_{s_ns_m}=3$, $n,m\in\{1,...,N\}$, $n\neq m$.}\label{F3}
\end{figure}

The average sum throughput of SUs as a function of
$\bar{\gamma}_{s_1p}$ is depicted in Fig. \ref{F1}, where
$\bar{\gamma}_{s_2p}=2$. As observed\footnote{Note that the SUs
interfere with each other in this paper, whereas the interference
between the SUs has been neglected in \cite{R22}.}, the SU sum
throughput decreases as $\bar{\gamma}_{s_1p}$ increases. This is
because $\bar{\gamma}_{s_2p}=2$ and hence, the PU throughput
degradation constraint is always active for the two SUs. A similar
plot for the case $\bar{\gamma}_{s_2p}=\bar{\gamma}_{s_1p}$ is
depicted in Fig. \ref{F8}. As observed, for
$\bar{\gamma}_{s_1p}<0.5$, $\bar{\gamma}_{s_1p}<0.25$ and
$\bar{\gamma}_{s_1p}<0.45$ respectively in the CRN with one SU,
centralized and decentralized cases, we have a different result. In
fact, because the interference power of SUs has little effect on the
PU receiver, initially the PU throughput degradation constraint is
not active and therefore $SU_{tx1}$ and $SU_{tx2}$ may utilize their
powers to maximize their own throughput. Note that the action
obtained for the SUs when $\bar{\gamma}_{s_1p}=0.25$ can not be used
for $\bar{\gamma}_{s_1p}<0.25$. In fact, notice that as
$\bar{\gamma}_{s_1p}$ and $\bar{\gamma}_{s_2p}$ increase, the
activity of the SUs causes more interference at the PU receiver and
leads to more ARQ retransmissions. In turn, this will make more IC
opportunities available at the SU receivers, thereby increasing the
SU sum throughput. On the other hand, since for PM already Known and
``No FIC'' the SUs assume that the PU messages are already known or
they do not apply IC, respectively, there is no benefit in
augmenting the ARQ retransmissions and therefore the performance is
constant for small $\bar{\gamma}_{s_1p}$ and $\bar{\gamma}_{s_2p}$,
until the constraint becomes active for $\bar{\gamma}_{s_1p}>0.5$,
$\bar{\gamma}_{s_1p}=\bar{\gamma}_{s_2p}>0.25$ and
$\bar{\gamma}_{s_1p}=\bar{\gamma}_{s_2p}>0.45$, respectively in the
CRN with one SU, centralized and decentralized cases; therefore,
above those values, the SU sum throughput decreases. As expected, in
the cognitive radio with two symmetric SUs and centralized scenario,
the PU throughput degradation constraint becomes active sooner than
in the cognitive radio with one SU, when increasing the SNR of the
channels from the SU transmitters to the PU receiver. A similar
observation can be made when
$\bar{\gamma}_{ps_1}=\bar{\gamma}_{ps_2}=2$ as depicted in Fig.
\ref{F4}. Our results, not shown here, confirm the same observation
for $N=3$ when compared with $N=2$. It is noteworthy that because
$\bar{\gamma}_{ps_1}=\bar{\gamma}_{ps_2}=2$ are neither strong
enough to be successfully decoded, nor so weak as to be considered
as small noise at the SU receivers, the SU sum throughput provided
by the centralized case suffers a higher performance loss with
respect to the upper bound compared with that in Fig. \ref{F8}. This
observation is clearly seen also in the next two figures, as
discussed later.
\begin{figure}[t!]
\centering
\includegraphics[width=3. in]{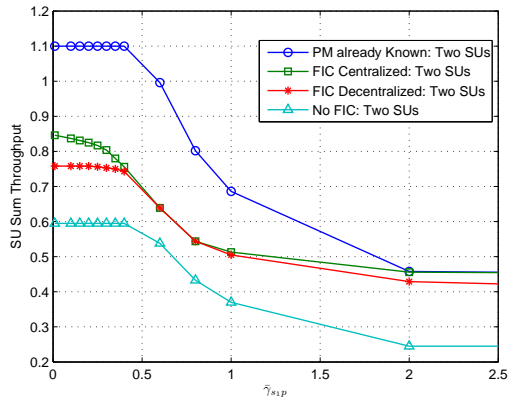}\caption{Average sum throughput of SUs with respect to $\bar{\gamma}_{s_1p}$. $\bar{\gamma}_{ps_1}=\bar{\gamma}_{ps_2}=5$, $\bar{\gamma}_{s_2p}=2$, $\bar{\gamma}_{pp}=10$, $\bar{\gamma}_{s_1s_1}=\bar{\gamma}_{s_2s_2}=5$, $\bar{\gamma}_{s_1s_2}=\bar{\gamma}_{s_2s_1}=3$ and
$\epsilon_{PU}=0.2$.}\label{F1}
\end{figure}

\begin{figure}[t!]
\centering
\includegraphics[width=3. in]{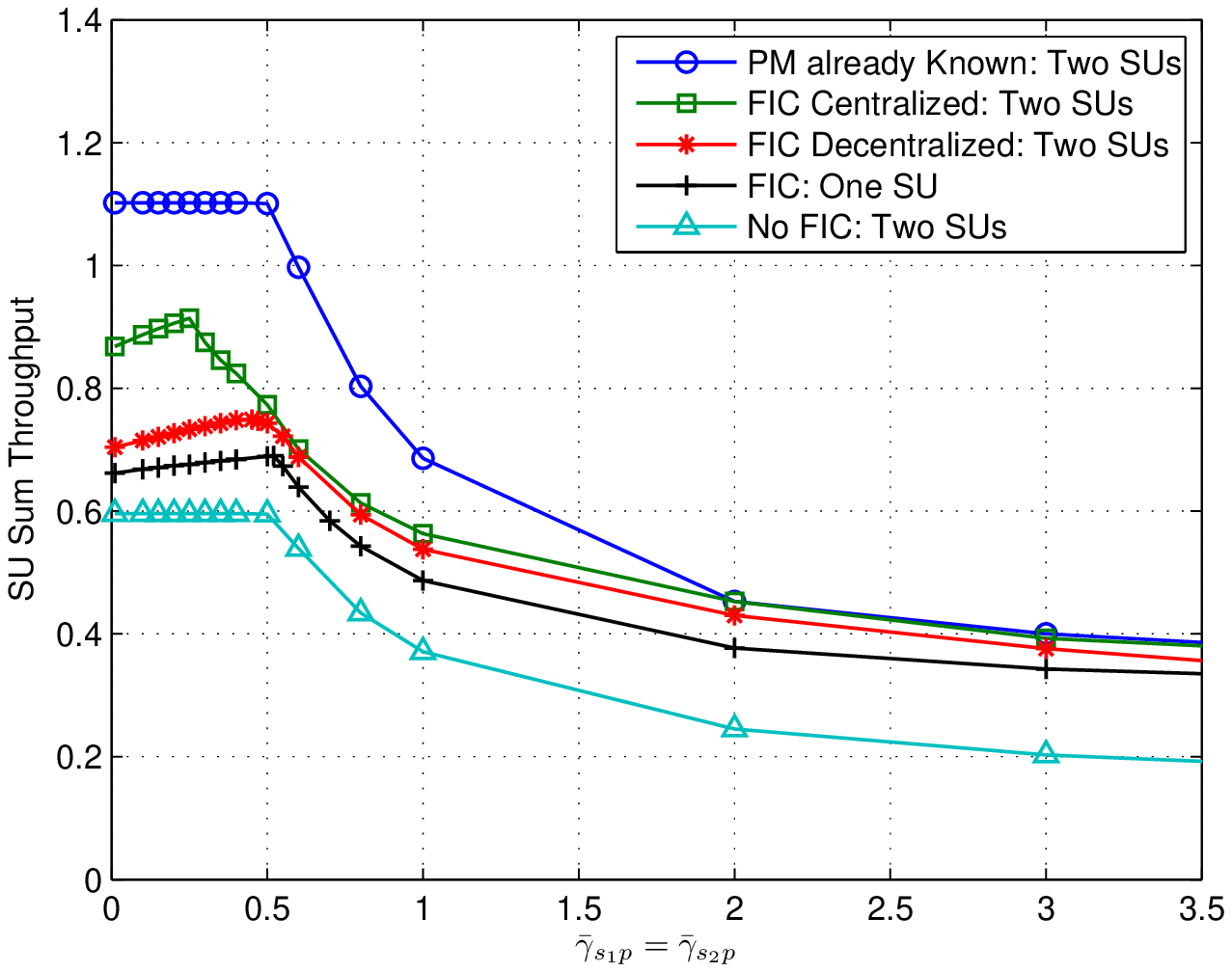}\caption{Average sum throughput of SUs with respect to $\bar{\gamma}_{s_1p}=\bar{\gamma}_{s_2p}$. $\bar{\gamma}_{ps_1}=\bar{\gamma}_{ps_2}=5$, $\bar{\gamma}_{pp}=10$, $\bar{\gamma}_{s_1s_1}=\bar{\gamma}_{s_2s_2}=5$, $\bar{\gamma}_{s_1s_2}=\bar{\gamma}_{s_2s_1}=3$ and
$\epsilon_{PU}=0.2$.}\label{F8}
\end{figure}

\begin{figure}[t!]
\centering
\includegraphics[width=3. in]{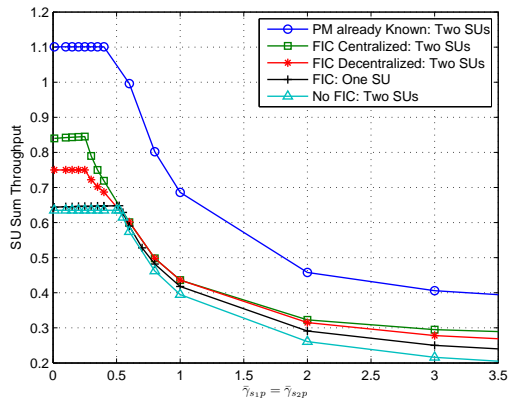}\caption{Average sum throughput of SUs with respect to $\bar{\gamma}_{s_1p}=\bar{\gamma}_{s_2p}$. $\bar{\gamma}_{ps_1}=\bar{\gamma}_{ps_2}=2$, $\bar{\gamma}_{pp}=10$, $\bar{\gamma}_{s_1s_1}=\bar{\gamma}_{s_2s_2}=5$, $\bar{\gamma}_{s_1s_2}=\bar{\gamma}_{s_2s_1}=3$ and
$\epsilon_{PU}=0.2$.}\label{F4}
\end{figure}

Figs. \ref{F2} and \ref{F9} show the average SU sum throughput with
respect to $\bar{\gamma}_{ps_1}$ for $\bar{\gamma}_{ps_2}=5$ and
$\bar{\gamma}_{ps_2}=\bar{\gamma}_{ps_1}$, respectively. Note that
$R^*_{s_1,a,\Phi=(U,\theta)}$ and $R^*_{s_2,a,\Phi=(\theta,U)}$
respectively depend on $\bar{\gamma}_{ps_1}$ and
$\bar{\gamma}_{ps_2}$. As expected, $\bar{\gamma}_{ps_1}$ does not
have any influence on the ``PM already Known'' scheme. This is
because in this scheme the PU message is previously known and can
always be canceled by the SU receiver in future retransmissions. It
is observed that for large enough values of $\bar{\gamma}_{ps_1}$,
the upper bound is achievable by the FIC scheme in the centralized
scenario. In fact, the SU receiver can successfully decode the PU
message, remove the interference and decode its corresponding
message. Note that the upper bound is computed in the centralized
scenario. The sum throughput is minimized at $\bar{\gamma}_{ps_1}=2$
in the CRN with one SU, centralized and decentralized cases, where
the PU message is neither strong enough to be successfully decoded,
nor weak to be considered as negligible. It is also evident that the
FIC scheme in Fig. \ref{F2} converges to the upper bound faster than
in Fig. \ref{F9}. The reason is that $\bar{\gamma}_{sp_1}$ and
$\bar{\gamma}_{sp_2}$ increase simultaneously in Fig. \ref{F9},
whereas the value of $\bar{\gamma}_{sp_2}$ is considered to be equal
to zero in Fig. \ref{F2}, resulting in no interference to the PU
receiver. It is also observed from Fig. \ref{F9} that a cognitive
radio with two symmetric SUs converges to the upper bound faster
than the network with one SU for large enough SNR of the channels
from the PU transmitter to SU receivers. This is because of the use
of the FIC scheme at the SU receivers. A similar behavior has been
observed in a CRN with $N=3$.

\begin{figure}[t!]
\centering
\includegraphics[width=3. in]{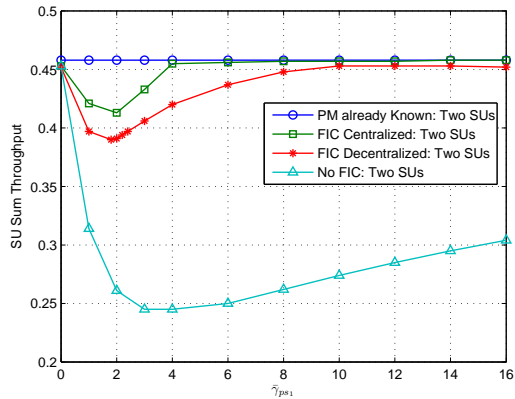}\caption{Average sum throughput of SUs with respect to $\bar{\gamma}_{ps_1}$. $\bar{\gamma}_{ps_2}=5$, $\bar{\gamma}_{s_1p}=\bar{\gamma}_{s_2p}=2$,
 $\bar{\gamma}_{pp}=10$, $\bar{\gamma}_{s_1s_1}=\bar{\gamma}_{s_2s_2}=5$, $\bar{\gamma}_{s_1s_2}=\bar{\gamma}_{s_2s_1}=3$ and $\epsilon_{PU}=0.2$. }\label{F2}
\end{figure}

\begin{figure}[t!]
\centering
\includegraphics[width=3. in]{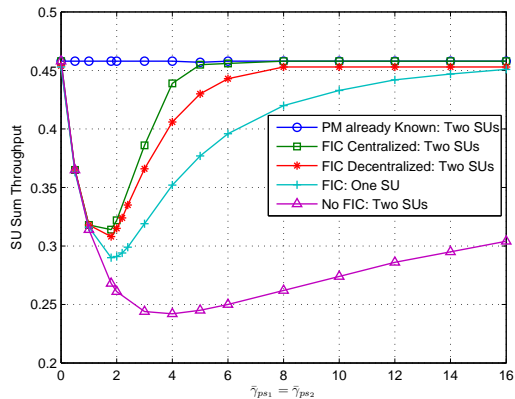}\caption{Average sum throughput of SUs with respect to $\bar{\gamma}_{ps_1}=\bar{\gamma}_{ps_2}$. $\bar{\gamma}_{s_1p}=\bar{\gamma}_{s_2p}=2$,
 $\bar{\gamma}_{pp}=10$, $\bar{\gamma}_{s_1s_1}=\bar{\gamma}_{s_2s_2}=5$, $\bar{\gamma}_{s_1s_2}=\bar{\gamma}_{s_2s_1}=3$ and $\epsilon_{PU}=0.2$. }\label{F9}
\end{figure}

\section{Extension to Decentralized Access Policy Design with Partially State Information}\label{SV}
In this section, we discuss a possible model for the decentralized
scenario when the PU message knowledge state is known partially for
the SUs in addition to the action being selected by the SU
independently of the other SUs. In Section \ref{S_DeC}, the PU
message knowledge state of each SU was assumed to be also known to
the other SUs, which makes the whole state of the system known to
all SUs. Now we assume that each user can only observe its own PU
message knowledge state. When there is an uncertainty about the
state of the system, the problem is called ``Distributed Partial
State Information MDP'' (DEC-PSI-MDP) which is a type of ``Partially
Observable MDP'' (DEC-POMDP). For a literature review on the
decentralized control of DEC-POMDP, the reader is referred to
\cite{R24}. In this model, the shared objective function is used
(here the SU Sum throughput) and the action is selected based on the
partial state observation at each SU. Because each secondary user is
unaware of the belief states of the other users, it is impossible
for each user to properly estimate the state of the system. Thus, a
DEC-POMDP can not be formulated as an MDP by introducing beliefs. It
can be shown that DEC-POMDP is nondeterministic exponential (NEXP)
complete even for two users \cite{R25} and, hence, only approximate
solutions can be applied \cite{R20}. Consideration of this type of
system is left as future work.

\section{Conclusion}\label{SVI}
In this paper, an optimal access policy for an arbitrary number of
cognitive secondary users was proposed, under a constraint on the
interference from the secondary users to the primary receiver.
Leveraging the inherent redundancy of the ARQ retransmissions
implemented by the PU, each SU receiver can cancel a successfully
decoded PU message in the following ARQ retransmissions, thereby
improving its own throughput. Both centralized and decentralized
scenarios were considered. In the first scenario, there is a
centralized unit which controls the access to the channel of all
SUs, to maximize the average sum throughput of the SUs under the
average PU throughput degradation constraint. In the decentralized
scenario, there exists no central unit and therefore each SU makes
an access decision independently of the other SUs, while the state
of the system is still assumed to be known to all secondary users.
In the centralized case, an upper bound was formulated and a close
form solution was provided. Our studies confirm that the centralized
and decentralized scenarios may be modeled as CMDP and MMDP and
therefore solved by linear programming. At the end, extension of the
problem to CRN with partial state information was discussed.

\appendices

\section{Proof of Proposition \ref{PR2}}\label{A1}
Define:
\begin{align}
d_i=Dec2Bin_N(i).({T}_{s_1,i,(K,...,K)},...,{T}_{s_N,i,(K,...,K)})\;\;\;i\in\{1,...,2^{N}-1\},
\end{align}
where ${T}_{s_n,i,(K,...,K)}$, $n\in\{1,...,N\}$ is given in
\eqref{E30}. A list of all situations is given here in detail for
$N=2$, and can be extended to an arbitrary $N$.
\begin{enumerate}
\item  $\lambda_1=0$ and $\lambda_2=0$. From \eqref{E92}, it is necessary to have
\begin{align}
d_i=0\;\;\;i\in\{1,2,3\}
\end{align}
Hence, this case is not acceptable.
\item $\mu_i=0,\;i\in\{1,2,3\}$. This case gives an SU sum
throughput equal to zero and hence does not provide the optimum
solution.

\item $\lambda_1>0$, $\lambda_2=0$,
$\mu_i>0,\mu_j=0,\mu_k=0$, $(i,j,k)\in\{(1,2,3),(2,1,3),(3,1,2)\}$.
It is observed from condition \eqref{E92} that $\frac{\partial
L}{\partial\mu_i}=0$, $\frac{\partial L}{\partial\mu_j}\leq0$ and
$\frac{\partial L}{\partial\mu_k}\leq0$. This occurs if
\begin{align}
&\frac{d_j}{{\rho_{p,j}-\rho_{p,0}}}\leq
\frac{d_i}{{\rho_{p,i}-\rho_{p,0}}}\\
&\frac{d_k}{{\rho_{p,k}-\rho_{p,0}}}\leq
\frac{d_i}{{\rho_{p,i}-\rho_{p,0}}}
\end{align}
Noting \eqref{E93} and \eqref{E94}, we have
$\mu_i({\rho_{p,i}-\rho_{p,0}})=\epsilon_{\omega}$ and $\mu_i\leq1$;
or equivalently
\begin{align}
\mu_i=\frac{\epsilon_{\omega}}{\rho_{p,i}-\rho_{p,0}}\leq1.
\end{align}
Thus, the resulting maximum SU sum throughput is equal to
$\frac{\epsilon_{\omega}}{\rho_{p,i}-\rho_{p,0}}{d_i}$.

\item $\lambda_1=0$, $\lambda_2>0$,
$\mu_i>0,\mu_j=0,\mu_k=0$, $(i,j,k)\in\{(1,2,3),(2,1,3),(3,1,2)\}$.
It is observed from condition \eqref{E92} that $\frac{\partial
L}{\partial\mu_i}=0$, $\frac{\partial L}{\partial\mu_j}\leq0$ and
$\frac{\partial L}{\partial\mu_k}\leq0$. This occurs if
\begin{align}
&d_k\leq d_i\\
&d_j\leq d_i.
\end{align}
Noting \eqref{E93} and \eqref{E94}, we have
$\mu_i({\rho_{p,i}-\rho_{p,0}})\leq\epsilon_{\omega}$ and $\mu_i=1$;
or equivalently
\begin{align}
\mu_i=1\leq\frac{\epsilon_{\omega}}{\rho_{p,i}-\rho_{p,0}}.
\end{align}
Thus, the resulting maximum SU sum throughput is equal to $d_i$.

\item $\lambda_1>0$, $\lambda_2>0$,
$\mu_i>0,\mu_j=0,\mu_k=0$, $(i,j,k)\in\{(1,2,3),(2,1,3),(3,1,2)\}$.
It is observed from condition \eqref{E92} that $\frac{\partial
L}{\partial\mu_i}=0$, $\frac{\partial L}{\partial\mu_j}\leq0$ and
$\frac{\partial L}{\partial\mu_k}\leq0$. This occurs if
\begin{align}
&d_k\geq d_i\;\;\;\mathrm{if}\;{{\rho_{p,k }-\rho_{p,0}}\geq{\rho_{p,i}-\rho_{p,0}}}\\
&d_k< d_i\;\;\;\mathrm{if}\;{{\rho_{p,k }-\rho_{p,0}}<{\rho_{p,i}-\rho_{p,0}}}\\
&d_j\geq
d_i\;\;\;\mathrm{if}\;{{\rho_{p,j}-\rho_{p,0}}\geq{\rho_{p,i}-\rho_{p,0}}}\\
&d_j<
d_i\;\;\;\mathrm{if}\;{{\rho_{p,j}-\rho_{p,0}}<{\rho_{p,i}-\rho_{p,0}}}.
\end{align}
Noting \eqref{E93} and \eqref{E94}, we have
\begin{align}
\mu_i=1=\frac{\epsilon_{\omega}}{\rho_{p,i}-\rho_{p,0}}.
\end{align}
Thus, the resulting maximum SU sum throughput is equal to $d_i$.

\item $\lambda_1>0$, $\lambda_2=0$,
$\mu_1>0,\mu_2>0,\mu_3>0$. It is observed from condition \eqref{E92}
that $\frac{\partial L}{\partial\mu_1}=\frac{\partial
L}{\partial\mu_2}=\frac{\partial L}{\partial\mu_3}=0$. This occurs
if
\begin{align}
\frac{d_1}{{\rho_{p,1}-\rho_{p,0}}}=\frac{d_2}{{\rho_{p,2}-\rho_{p,0}}}=\frac{d_3}{{\rho_{p,3}-\rho_{p,0}}}.
\end{align}
 Noting \eqref{E93} and \eqref{E94},
$\mu_1({\rho_{p,1}-\rho_{p,0}})+\mu_2({\rho_{p,2}-\rho_{p,0}})+\mu_3({\rho_{p,3}-\rho_{p,0}})=\epsilon_{\omega}$
and $\mu_1+\mu_2+\mu_3\leq1$. These two conditions impose that
\begin{align}
\epsilon_{\omega}\leq {\rho_{p,3}-\rho_{p,0}}.
\end{align}
Thus, the resulting maximum SU sum throughput is equal to
$\frac{\epsilon_{\omega}{d_1}}{{\rho_{p,1}-\rho_{p,0}}}$.

\item $\lambda_1=0$, $\lambda_2>0$,
$\mu_1>0,\mu_2>0,\mu_3>0$. It is observed from condition \eqref{E92}
that $\frac{\partial L}{\partial\mu_1}=\frac{\partial
L}{\partial\mu_2}=\frac{\partial L}{\partial\mu_3}=0$. This occurs
if
\begin{align}
d_1=d_2=d_3.
\end{align}
 Noting \eqref{E93} and \eqref{E94},
$\mu_1({\rho_{p,1}-\rho_{p,0}})+\mu_2({\rho_{p,2}-\rho_{p,0}})+\mu_3({\rho_{p,3}-\rho_{p,0}})\leq\epsilon_{\omega}$
and $\mu_1+\mu_2+\mu_3=1$. The conditions impose that
\begin{align}
\operatorname{min}{(\rho_{p,1}-\rho_{p,0},\rho_{p,2}-\rho_{p,0})}\leq\epsilon_{\omega}.
\end{align}
Thus, the resulting maximum SU sum throughput is equal to $d_1$.

\item $\lambda_1>0$, $\lambda_2>0$,
$\mu_1>0,\mu_2>0,\mu_3>0$. It is observed from condition \eqref{E92}
that $\frac{\partial L}{\partial\mu_1}=\frac{\partial
L}{\partial\mu_2}=\frac{\partial L}{\partial\mu_3}=0$. This occurs
if
\begin{align}
&d_i\leq d_3\;\;\;i\in\{1,2\}\\
&d_1\geq d_2\;\;\;\mathrm{if}\;{{\rho_{p,1 }-\rho_{p,0}}\geq{\rho_{p,2}-\rho_{p,0}}}\\
&d_1< d_2\;\;\;\mathrm{if}\;{{\rho_{p,1 }-\rho_{p,0}}<{\rho_{p,2}-\rho_{p,0}}}\\
&\frac{d_3}{\rho_{p,3}-\rho_{p,0}}<\operatorname{min} \{
\frac{d_1}{\rho_{p,1}-\rho_{p,0}},\frac{d_2}{\rho_{p,2}-\rho_{p,0}}\}\\
&\frac{d_1}{\rho_{p,1}-\rho_{p,0}}\geq\frac{d_2}{\rho_{p,2}-\rho_{p,0}}\;\;\;\mathrm{if}\;{{\rho_{p,1
}-\rho_{p,0}}\leq{\rho_{p,2}-\rho_{p,0}}}\\
&\frac{d_1}{\rho_{p,1}-\rho_{p,0}}<\frac{d_2}{\rho_{p,2}-\rho_{p,0}}\;\;\;\mathrm{if}\;{{\rho_{p,1
}-\rho_{p,0}}>{\rho_{p,2}-\rho_{p,0}}}
\end{align}
Noting \eqref{E93} and \eqref{E94},
$\mu_1({\rho_{p,1}-\rho_{p,0}})+\mu_2({\rho_{p,2}-\rho_{p,0}})+\mu_3({\rho_{p,3}-\rho_{p,0}})=\epsilon_{\omega}$
and $\mu_1+\mu_2+\mu_3=1$. The conditions impose that
\begin{align}
&\operatorname{min}{(\rho_{p,1}-\rho_{p,0},\rho_{p,2}-\rho_{p,0})}\leq\epsilon_{\omega}\leq\operatorname{max}{(\rho_{p,1}-\rho_{p,0},\rho_{p,2}-\rho_{p,0})}\\
&\epsilon_{\omega} \le{\rho_{p,3}-\rho_{p,0}}.
\end{align}
 Thus, the resulting maximum SU sum throughput is equal or lower
 than
$\epsilon_{\omega}\operatorname{max}{(\frac{d_1}{\rho_{p,1}-\rho_{p,0}},\frac{d_2}{\rho_{p,2}-\rho_{p,0}})}$
and the equality is achieved when
$\frac{d_1}{\rho_{p,1}-\rho_{p,0}}=\frac{d_2}{\rho_{p,2}-\rho_{p,0}}$.

\item $\lambda_1>0$, $\lambda_2=0$,
$\mu_i=0,\mu_j>0,\mu_k>0$, $(i,j,k)\in\{(1,2,3),(2,1,3),(3,1,2)\}$.
It is observed from condition \eqref{E92} that $\frac{\partial
L}{\partial\mu_i}\leq0$, $\frac{\partial L}{\partial\mu_j}=0$ and
$\frac{\partial L}{\partial\mu_k}=0$. This occurs if
\begin{align}
\frac{d_i}{{\rho_{p,i}-\rho_{p,0}}}\leq\frac{d_j}{\rho_{p,j}-\rho_{p,0}}=\frac{d_k}{\rho_{p,k}-\rho_{p,0}}.
\end{align}
 Noting \eqref{E93} and \eqref{E94},
$\mu_j({\rho_{p,j}-\rho_{p,0}})+\mu_k({\rho_{p,k}-\rho_{p,0}})=\epsilon_{\omega}$
and $\mu_j+\mu_k\leq1$. The conditions impose that
\begin{align}
\epsilon_{\omega}\leq\operatorname{max}{(\rho_{p,j}-\rho_{p,0},\rho_{p,k}-\rho_{p,0})}.
\end{align}
and the resulting maximum SU sum throughput is equal to
$\frac{\epsilon_{\omega}{d_j}}{{\rho_{p,j}-\rho_{p,0}}}$.

\item $\lambda_1=0$, $\lambda_2>0$,
$\mu_i=0,\mu_j>0,\mu_k>0$, $(i,j,k)\in\{(1,2,3),(2,1,3),(3,1,2)\}$.
It is observed from condition \eqref{E92} that $\frac{\partial
L}{\partial\mu_i}\leq0$, $\frac{\partial L}{\partial\mu_j}=0$ and
$\frac{\partial L}{\partial\mu_k}=0$. This occurs if
\begin{align}
d_i\leq{d_j}=d_k.
\end{align} Noting \eqref{E93} and \eqref{E94},
$\mu_j({\rho_{p,j}-\rho_{p,0}})+\mu_k({\rho_{p,k}-\rho_{p,0}})\leq\epsilon_{\omega}$
and $\mu_j+\mu_k=1$.  The conditions impose that
\begin{align}
\operatorname{min}{(\rho_{p,j}-\rho_{p,0},\rho_{p,k}-\rho_{p,0})}\leq\epsilon_{\omega}.
\end{align}
 The resulting maximum SU sum throughput is equal to
$d_j$.

\item $\lambda_1>0$, $\lambda_2>0$,
$\mu_i=0,\mu_j>0,\mu_k>0$, $(i,j,k)\in\{(1,2,3),(2,1,3),(3,1,2)\}$.
It is observed from condition \eqref{E92} that $\frac{\partial
L}{\partial\mu_i}\leq0$, $\frac{\partial L}{\partial\mu_j}=0$ and
$\frac{\partial L}{\partial\mu_k}=0$. This occurs if
\begin{align}
&d_j\geq d_k\;\;\;\mathrm{if}\;{{\rho_{p,j }-\rho_{p,0}}\geq{\rho_{p,k}-\rho_{p,0}}}\\
&d_j< d_k\;\;\;\mathrm{if}\;{{\rho_{p,j }-\rho_{p,0}}<{\rho_{p,k}-\rho_{p,0}}}\\
&\frac{d_j}{\rho_{p,j}-\rho_{p,0}}\geq\frac{d_k}{\rho_{p,k}-\rho_{p,0}}\;\;\;\mathrm{if}\;{{\rho_{p,j }-\rho_{p,0}}\leq{\rho_{p,k}-\rho_{p,0}}}\\
&\frac{d_j}{\rho_{p,j}-\rho_{p,0}}<\frac{d_k}{\rho_{p,k}-\rho_{p,0}}\;\;\;\mathrm{if}\;{{\rho_{p,j }-\rho_{p,0}}>{\rho_{p,k}-\rho_{p,0}}}\\
&d_i\geq d_j\;\;\;\mathrm{if}\;{{\rho_{p,i }-\rho_{p,0}}\geq{\rho_{p,j}-\rho_{p,0}}}\\
&d_i< d_j\;\;\;\mathrm{if}\;{{\rho_{p,i }-\rho_{p,0}}<{\rho_{p,j}-\rho_{p,0}}}\\
&d_i\geq d_k\;\;\;\mathrm{if}\;{{\rho_{p,i
}-\rho_{p,0}}\geq{\rho_{p,k}-\rho_{p,0}}}\\
&d_i< d_k\;\;\;\mathrm{if}\;{{\rho_{p,i
}-\rho_{p,0}}<{\rho_{p,k}-\rho_{p,0}}}
\end{align}
 Noting \eqref{E93} and \eqref{E94},
$\mu_j({\rho_{p,j}-\rho_{p,0}})+\mu_k({\rho_{p,k}-\rho_{p,0}})=\epsilon_{\omega}$
and $\mu_j+\mu_k=1$.  The conditions impose that
\begin{align}
\operatorname{min}{(\rho_{p,j}-\rho_{p,0},\rho_{p,k}-\rho_{p,0})}\leq\epsilon_{\omega}\leq\operatorname{max}{(\rho_{p,j}-\rho_{p,0},\rho_{p,k}-\rho_{p,0})}.
\end{align}
 The resulting maximum SU sum throughput is equal or lower than
 $\epsilon_{\omega}\operatorname{max}{(\frac{d_j}{\rho_{p,j}-\rho_{p,0}},\frac{d_k}{\rho_{p,k}-\rho_{p,0}})}$
 and the equality is achieved when
 $\frac{d_j}{\rho_{p,j}-\rho_{p,0}}=\frac{d_k}{\rho_{p,k}-\rho_{p,0}}$.
\end{enumerate}
Noting items 1 to 11, it is observed that items $3$, $4$, $6$, $7$,
$9$, $10$ provide optimum solutions and hence, the optimum access
policy and SU Sum throughput can be summarized in \eqref{E91} and
\eqref{E28} respectively. Thus, the proof is complete.

\bibliographystyle{Ieeetr}
\bibliography{IEEEabrv,Ref}

\end{document}